\documentclass[11pt]{article}
\usepackage{graphicx}
\usepackage{graphics,amsmath,amssymb,amsthm,accents,amsfonts}
\usepackage{epic}
\usepackage{amsmath}
\usepackage{color}
\usepackage{cite}
\makeatletter

\newcommand{\Rmnum}[1]{\expandafter\@slowromancap\romannumeral #1@}
\makeatother \textwidth=16truecm \textheight=22.5truecm
\def \h#1{\widehat{#1}}
\def \t#1{\widetilde{#1}}
\def \b#1{\overline{#1}}
\def \c#1{\accentset{\circ}{#1}}
\def \th#1{\widehat{\widetilde{#1}}}
\def \hb#1{{\widehat{\overline{#1}}}}
\def \bh#1{{\widehat{\overline{#1}}}}

\def \tb#1{\widetilde{\overline{#1}}}
\def \bt#1{\widetilde{\overline{#1}}}
\def \dh#1{\underaccent{\hat}{#1}}
\def \db#1{\underaccent{\bar}{#1}}
\def\dwh{\underaccent{{\cc@style\widehat{\mskip10mu}}}}
\def \dt#1{\underaccent{\tilde}{#1}}
\def \dth#1{\underaccent{\hat}{\underaccent{\tilde}{#1}}}
\def \dtb#1{\underaccent{\bar}{\underaccent{\tilde}{#1}}}
\def \dhb#1{\underaccent{\bar}{\underaccent{\hat}{#1}}}

\makeatother \numberwithin{equation}{section}
\newtheorem{prop}{Proposition}
\allowdisplaybreaks

\begin{document}
%\begin{CJK*}{GBK}{kai}
\title {Solutions to the non-autonomous ABS lattice equations: Casoratians and bilinearization}
\author{Ying Shi\footnote{Corresponding author. E-mail: shiying0707@shu.edu.cn},~~
Da-jun Zhang\footnote{E-mail: djzhang@staff.shu.edu.cn},~~Song-lin Zhao\\
{\small\it Department of Mathematics,
 Shanghai University, Shanghai 200444,  P.R. China}}
\maketitle
\date{}

\begin{abstract}
In the paper non-autonomous H1, H2, H3$_\delta$ and Q1$_\delta$ equations in the ABS list
are bilinearized. Their solutions are derived in Casoratian form. We also list out some Casoratian shift formulae which are used to verify Casoratian solutions.
\end{abstract}


%%%%%%%%%%%%%%%%%%%%%%%%%%%%%%%%%%%%%%%%%%%%%%%%%%%%%%%%%%%%%%%%%%%%%%%%%%%%%%%%%%%%%%%%%%%%%%%%%%%%%%%%%%%%%%%%%

\section{Introduction}

The ``discrete integrable systems'' has been a popular topic  and is
still drawing more and more attention. Particularly in the recent ten
years it received much progress. The property of multidimensional
consistency\cite{Nijhoff-Glasgow-2001,BS-IMRN-2002,ABS-CMP-2003}
provides an approach to investigate integrability for the discrete
systems defined on an elementary quadrilateral:
\begin{equation}\label{eq-Q}
Q(u_{n,m},u_{n+1,m},u_{n,m+1},u_{n+1,m+1};p,q)=0,
\end{equation}
where $p,q$ are spacing parameters. The ABS list\cite{ABS-CMP-2003}
contains all quadrilateral lattice equations of the above form which
are consistent-around-the-cube(CAC). In the ABS list the spacing
parameters $p,q$ can either constants or functions $p_n$ and $q_m$,
which corresponds to autonomous case or non-autonomous case, respectively.
In general an autonomous system means a differential/difference model with constant coefficients while a
non-autonomous one  means the model has coefficients  varying with
independent variables but the model can not be transformed back to an
autonomous one. Obviously, the ABS lattice equations themselves are
automatically non-autonomous in the sense of taking $p=p_n, q=q_m$,
and this non-autonomous case still keeps the CAC property.

Integrable non-autonomous systems have its own importance. It is
known that most of discrete Painlev\'e equations are non-autonomous
ordinary difference equations. In addition, for integrable
non-autonomous forms of partial difference equations, their
reductions usually lead to integrable non-autonomous mappings, which
quite often are discrete Painlev\'e equations. To get a
non-autonomous version of a partial difference equation, taking
\eqref{eq-Q} as an example, one can replace constant lattice
parameters $(p,q)$ by $(p_{n,m},q_{n,m})$, but the integrability
should be kept. Many criterions, such as singularity confinement,
conservation laws and algebraic entropy, have been used to check
integrability for the non-autonomous systems, for both ordinary and
partial difference
cases\cite{Grammaticos-PRL,Papag-PLA,Hydon-JMAA,GRamani-LMP}.
Besides, it is also possible to deautonomise a discrete bilinear
system if it contains spacing parameters. With suitable
deautonomisation the obtained non-autonomous bilinear systems admit
$N$-soliton solutions expressed through deformed discrete
exponential functions, (see \cite{Kajiwara-Ohta-1,Kajiwara-Ohta-2} as examples).

Recently, many solving approaches have been developed to find
solutions for autonomous lattice equations in the ABS
list\cite{Q4,Q3,NAJ-PartI,HZ-PartII,Atkinson-CMP,Nijhoff-IMRN,BJ-IST,ZH-H1,SZ-SIGMA-2010}.
In \cite{HZ-PartII} the H1, H2, H3$_\delta$ and Q1$_\delta$
equations in the ABS list were bilinearized and their solutions were
derived in Casoratian form. In the present paper we will repeat the
treatment of \cite{HZ-PartII} to get bilinear forms as well as
solutions in Casoratian form for some non-autonomous ABS lattice
equations. As we have mentioned before, the ABS lattice equations
with spacing parameters $(p_n,q_m)$ are automatically non-autonomous
and still CAC. Their integrable aspects are also double checked by
singularity confinement and algebraic entropy
approaches\cite{GRamani-LMP}.

The paper is organized as follows. Section 2 contains some basic
notations for discrete systems and Casoratians and a list of
non-autonomous ABS lattice equations. In Section 3 the
non-autonomous H1, H2, H3$_\delta$ and Q1$_\delta$ equations are
bilinearized and their solutions are derived in Casoratian form. The
Appendix contains a collection of Casoratian formulae of
non-autonomous case.

%%%%%%%%%%%%%%%%%%%%%%%%%%%%%%%%%%%%%%%%%%%%%%%%%%%%%%%%%%%%%%%%%%%%%%%%%%%%%%%%%%%%%%%%

\section{Preliminaries}

Conventionally, we use tilde/hat notations to express the shifts in
$n$/$m$ directions, for example,
\begin{equation*}\label{Notations}
    u=u_{n,m},~~\t u=u_{n+1,m},~~\dt u=u_{n-1,m},~~\h u=u_{n,m+1},~~\dh u=u_{n,m-1},~~\th u=u_{n+1,m+1}.
\end{equation*}
By these notations the lattice equation \eqref{eq-Q} is rewritten as
\begin{equation}\label{PDE}
    \mathcal{Q}(u,\t u, \h u, \th u; p,q)=0.
\end{equation}

The non-autonomous ABS list is as
follows\cite{ABS-CMP-2003,GRamani-LMP}:
\begin{subequations}\label{H/Q-list}
\begin{align}
\label{H1} \mbox{H1}:~~&(u-\h{\t u}) (\t u-\h
u)+q_m-p_n=0,\\
\label{H2} \mbox{H2}:~~&(u-\h{\t u}) (\t u-\h u)+ (q_m-p_n) (u+\t
u+\h u+\h{\t
u})+q_m^2-p_n^2 =0,\\
\label{H3} \mbox{H3}_\delta:~~& p_n(u \t u+\h u \h{\t u})- q_m(u \h
u+\t u \h{\t
u})+\delta (p_n^2-q_m^2) =0,\\
\label{Q1} \mbox{Q1}_\delta:~~& p_n  (u-\h u) (\t u-\h{\t u})- q_m
(u-\t u) (\h
u-\h{\t u})+\delta^2 p_n  q_m (p_n-q_m) =0,\\
\mbox{Q2}:~~&p_n  (u-\h u) (\t u-\h{\t u})- q_m (u-\t u) (\h u-\h{\t
u})+ p_n  q_m (p_n -q_m) (u+\t u+\h u+\h{\t u})\nonumber\\
~~&~~-p_n q_m (p_n -q_m) (p_n ^2-p_n  q_m+q_m^2)
=0,\label{Q2}\\
\mbox{Q3}_\delta:~~&\sin(p_n+q_m)(u \h{\t u}+\t u \h u)-\sin p_n(u
\t u+\h u \h{\t u}) -\sin q_m(u \h u +\t u \h{\t u})\nonumber\\
~~&~~+\delta^2 \sin p_n\sin
q_m\sin(p_n+q_m)=0,\label{Q3}\\
\mbox{Q4}:~~&\text{sn}(p_n+q_m)(u \h{\t u} + \t u \h
u)-\text{sn}p_n(u \t u + \h u \h{\t u} )-\text{sn}q_m(u \h u + \t u
\h{\t u})\nonumber\\
~~&~~+\text{sn}p_n\text{sn}q_m\text{sn}(p_n+q_m)(1+k^2u\t
u\h u\th u)=0,\label{Q4}
\end{align}
\end{subequations}
where $\delta$ is a constant, $p_n=p(n)$ and $q_m=q(m)$ are
arbitrary non-zero functions of discrete variables $n$ and $m$,
respectively.

Here the forms of Q3$_\delta$ and Q4 are in accordance with the
autonomous version via the parametrization introduced by Hietarinta
\cite{Hietari-JNMP}. We omit A1$_\delta$ and A2 from the above list
because of the equivalence between A1$_\delta$ and Q1$_\delta$ by
$u\to(-1)^{n+m}u$, as well as A2 and Q3$_{\delta=0}$ by $u\to
u^{(-1)^{n+m}}$.

The discrete version of Wronskian is Casoratian, which is a
determinant of the Casorati matrix:
\begin{subequations}\label{psi-T}
\begin{equation}
  \label{eq:C-gen}
  f=|\psi(n,m,l_1),\psi(n,m,l_2),\cdots,\psi(n,m,l_N)|=|l_1,l_2,\cdots,l_N|,
\end{equation}
where the basic column vector is
\begin{equation}\label{psi-vector}
    \psi
    (n,m,l)=(\psi_1(n,m,l),\psi_2(n,m,l),\cdots,\psi_N(n,m,l))^T,
\end{equation}
\end{subequations}
and the shifts are in $l$ direction. Using the standard short-hand
notations\cite{Freeman-Nimmo-KP}, we list the following often-used
$N$th-order Casoratians
\begin{equation*}
|\h{N-1}|=|0,1,\cdots,N-1|,~~
|\h{N-2},N|=|0,1,\cdots,N-2,N|,~~|-1,\t{N-1}|=|-1,1,2,\cdots,N-1|.
\end{equation*}

As in \cite{HZ-PartII}, since in \eqref{psi-T} there are three
direction variables, say $n,m$ and $l$, one can introduce the
operators $E^\nu$ $(\nu=1,2,3)$ by
\begin{equation}\label{C-gen0}
    E^1\psi=\widetilde{\psi}=\psi(n+1,m,l),~~E^2\psi=\widehat{\psi}=\psi(n,m+1,l),~~E^3\psi=\bar{\psi}=\psi(n,m,l+1),
\end{equation}
then define a Casoratian w.r.t $E^\nu$-shift,
\begin{equation}\label{C-gen1}
    |\widehat{N-1}|_{[\nu]}=|\psi,E^\nu\psi,(E^\nu)^2\psi,\cdots,(E^\nu)^{N-1}\psi|,~~(\nu=1,2,3).
\end{equation}
For these Casoratians we have
\begin{prop}\label{P:2.1}
The Casoratians
\begin{equation}
|\widehat{N-1}|_{_{[1]}}=|\widehat{N-1}|_{_{[2]}}=|\widehat{N-1}|_{_{[3]}},
\label{cas-iden}
\end{equation}
if their column vector $\psi(n,m,l)$ satisfies the relations
\begin{equation}\label{Prop-1}
\alpha_n\psi=\b \psi-\t \psi,~~ \beta_m\psi=\b \psi-\h \psi,
\end{equation}
where $\alpha_n $ and $\beta_m$ are arbitrary functions of discrete
variables $n$ and $m$, respectively, i.e. $\alpha_n=\alpha(n)$,
$\beta_m=\beta(m)$.
\end{prop}

\begin{proof}
By the definition of $E^{\nu}$ in \eqref{C-gen0}£¬the relations
\eqref{Prop-1} can be rewritten as
\begin{equation}\label{Prop-3}
E^3\psi=(E^1+\alpha_n)\psi, ~~E^3\psi=(E^2+\beta_m)\psi,
\end{equation}
from which one has
\begin{subequations}\label{Prop-5}
\begin{eqnarray}
(E^3)^k=(E^1+\alpha_n)^k=(E^1)^k+\sum_{j=1}^k\sum_{\substack{l_j=0 \\
l_j\leq
l_{j+1}}}^{k-j}\prod_{i=1}^{j}\alpha_{n+l_i}(E^1)^{k-j},~~k=1,2,\cdots,N-1,\\
\label{Prop-1-5}
(E^3)^k=(E^2+\beta_m)^k=(E^2)^k+\sum_{j=1}^k\sum_{\substack{l_j=0 \\
l_j\leq
l_{j+1}}}^{k-j}\prod_{i=1}^{j}\beta_{m+l_i}(E^2)^{k-j},~~k=1,2,\cdots,N-1
.\label{prop-2-5}
\end{eqnarray}
\end{subequations}
Then, substituting them into \eqref{C-gen1} one can easily obtain
\eqref{cas-iden}.
\end{proof}
This Proposition will bring more flexibility for Casoratian
verifications. Besides, for later convenience, we give the following
Laplace expansion property\cite{Freeman-Nimmo-KP}:
\begin{prop}\label{P:2.2}
Suppose that $\mathbf{B}$ is a $N\times(N-2)$ matrix, and
\textbf{a},\textbf{b},\textbf{c},\textbf{d} are $N$th-order column
vectors, then
\begin{equation}\label{}
    |\mathbf{B},\mathbf{a},\mathbf{b}||\mathbf{B},\mathbf{c},\mathbf{d}|
    -|\mathbf{B},\mathbf{a},\mathbf{c}||\mathbf{B},\mathbf{b},\mathbf{d}|+|\mathbf{B},\mathbf{a},\mathbf{d}||\mathbf{B},\mathbf{b},\mathbf{c}|=0.
\end{equation}
\end{prop}
%%%%%%%%%%%%%%%%%%%%%%%%%%%%%%%%%%%%%%%%%%%%%%%%%%%%%%%%%%%%%%%%%%%%%%%%%%%%%%%%%%%%%%%%%%%%%%%%%%%%%%%%%%%%%%%%%%%%%%%

\section{Bilinearization and Casoratian solutions}

In the following we derive bilinear forms and Casoratian solutions
for the non-autonomous H1, H2, H3$_\delta$ and Q1$_\delta$ models in
the non-autonomous ABS list \eqref{H/Q-list}. Singularity
confinement might provide a possible transformation to connect
discrete integrable systems with their bilinear forms, (see
\cite{sc-bil-1,sc-bil-2} as examples), but here we will roughly use
the same transformations as for the autonomous lattice
equations\cite{HZ-PartII}. It then turns out that these
non-autonomous lattice equations can share the same bilinear forms
with those autonomous ones except changing the spacing parameters
accordingly.

To get Casoratian solutions one needs to use deformed discrete
exponential functions and develop corresponding Caosratian shift
formulae, which we have listed in Appendix.

\subsection{Non-autonomous H1 equation}

We note that the non-autonomous H1 has been solved in
\cite{Kajiwara-Ohta-1} through bilinear approach and the Casoratian
solutions were given, but here we give a more generalized result.

With the parametrization
\begin{equation}\label{Parame}
    p_n=c-a_n^2,~~q_m=c-b_m^2,~~(c~\mathrm{is~ an~arbitrary~ constant}),
\end{equation}
and through the transformation
\begin{equation}
\label{H1-Trans}
    u_{n,m}=\frac{g_{n,m}}{f_{n,m}}-\sum^{n-1}_{i=n_{_0}}a_i-\sum^{m-1}_{j=m_{_0}}b_j-\gamma,~~
    (\gamma~\mathrm{ is~ an~arbitrary~ constant}),
\end{equation}
the non-autonomous H1 \eqref{H1} is bilinearized by
\begin{subequations}\label{Bili-H1}
\begin{align}
&\mathcal{H}_1\equiv(\h g\t f-\t g \h f)+(a_n-b_m)(\h f \t f-f\h{\t f})=0,\label{Bili-H1-1}\\
&\mathcal{H}_2\equiv(g \h{\t f}-\h{\t g}f)+(a_n+b_m)(f\h{\t f}-\h
f\t f)=0,\label{Bili-H1-2}
\end{align}
\end{subequations}
where in the transformation \eqref{H1-Trans}  $n_0, m_0$ are
arbitrary integers.
The connection between \eqref{H1} and \eqref{Bili-H1} is
\[
-\bigl[\mathcal{H}_1+(a_n-b_m)f\th f\bigr]
\bigl[\mathcal{H}_2+(a_n+b_m)\h f\t f\bigr]/(f\h f\t f \th f)
+(a_n^2-b_m^2) \equiv \mathrm{H1},
\]
which is the same relation as the autonomous one\cite{HZ-PartII}.

Solutions to the bilinear equations \eqref{Bili-H1} can be given by
\begin{prop}
\label{P:3.1-H1}
The Casoratians
\begin{equation}
f(\psi)=|\h{N-1}|_{_{[3]}},~~g(\psi)=|\h{N-2},N|_{_{[3]}},
\label{fg-H1}
\end{equation}
solve the non-autonomous bilinear equations \eqref{Bili-H1}, if the
column vector $\psi(n,m,l)$ symmetrically\footnote{ Here the
symmetric property between pairs $(n, a_n)$ and $(m,b_m)$ means, for
example, once we have \eqref{H1-ditui}, at the same time we have
\[
b_{m-1}\dh\psi=\psi-\overline{\dh\psi},
\]
and
\[\psi=A_{[n]}\omega,~~a_{n}\t \omega=\omega+\t{\b\omega}.
\]
}, in terms of the pairs $(n, a_n)$ and $(m,b_m)$, satisfies the
shift relations
\begin{subequations}\label{H1-ditui}
\begin{align}
&a_{n-1}\dt\psi=\psi-\overline{\dt\psi},\label{H1-ditui-1}\\
&\psi=A_{[m]}\phi,~~b_{m}\h \phi=\phi+\h{\b\phi},\label{H1-ditui-2}
\end{align}
\end{subequations}
\noindent where $\phi(n,m,l)$ is an auxiliary vector, the $N\times
N$ transform matrix $A_{[m]}$ is invertible, and the subscript $[m]$
specially means $A_{[m]}$ only depends on $m$ but is independent of
$(n,l)$.
\end{prop}

\begin{proof}

We prove $\mathcal{H}_1 $ in its down-tilde-hat
version $\mathcal{\dth H}_1$:
\begin{equation}\label{dth-Bili-H1-1}
    \mathcal{\dth H}_1\equiv(\dt g\dh f-\dh g \dt f)+(a_{n-1}-b_{m-1})(\dt f \dh f-\dth f
    f).
\end{equation}
Using the formulae given in appendix A with $c=0$. In
\eqref{dth-Bili-H1-1} $f=|\widehat{N-1}|_{_{[3]}}$, $\dth f$, $\dh
f$, $\dt g+a_{n-1}\dt f$, $\dt f$ and $\dh g+ b_{m-1}\dh f$ are
\eqref{2-1}, \eqref{6-2}, \eqref{3-1}, \eqref{1-2} and \eqref{6-15},
respectively, we have
\[
\begin{array}{rl}
\mathcal{\dth H}_1\equiv&-(a_{n-1}-b_{m-1})\dth f
    f+\dh f(\dt g+a_{n-1}\dt f)-\dt f(\dh g+b_{m-1}\dh f)\\
    =&-a_{n-1}^{-N+2}~b_{m-1}^{-N+2}~[|\h{N-3},~\psi(N-2),~\psi(N-1)|_{_{[3]}}\cdot
    |\h{N-3},~\dh\psi(N-2),~\dt\psi(N-2)|_{_{[3]}}\\
    &-|\h{N-3},~\psi(N-2),~\dh\psi(N-2)|_{_{[3]}}\cdot
    |\h{N-3},~\psi(N-1),~\dt\psi(N-2)|_{_{[3]}}\\
    &+|\h{N-3},~\psi(N-2),~\dt\psi(N-2)|_{_{[3]}}\cdot
    |\h{N-3},~\psi(N-1),~\dh\psi(N-2)|_{_{[3]}}]\\
    =&0,
\end{array}
\]
where we have made use of Proposition \ref{P:2.2} in which
$\textbf{B}=(\h {N-3})$, $(\textbf{a}, \textbf{b},\textbf{c},
\textbf{d})=(\psi(N-2), \psi(N-1), \dh\psi(N-2), \dt\psi(N-2))$.

Next, we prove the down-tilde version of $\mathcal{H}_2$, which is
\begin{equation}\label{dt-Bili-H1-2}
    \mathcal{\dt H}_2\equiv(\h f\dt g-\h g\dt f)+(a_{n-1}+b_{m})(\h f
    \dt f-f\h{\dt f}).
\end{equation}
In \eqref{dt-Bili-H1-2} we take $f=|\widehat{N-1}|_{_{[3]}}$, and for
$\h{\dt f}$, $\dt f$, $\h g- b_{m}\h f$, $\h f$ and $\dt
g+a_{n-1}\dt f$ we use \eqref{2-3}, \eqref{1-2}, \eqref{3-2},
\eqref{6-19} and \eqref{3-1}, with $c=0$, respectively. Then we have
\begin{equation*}
\begin{array}{rl}
\mathcal{\dt H}_2\equiv&-(a_{n-1}+b_{m})f\h{\dt f}-\dt f(\h g-b_{m}\h f)+\h f(\dt g+a_{n-1}\dt f)\\
    =&-a_{n-1}^{-N+2}~b_{m}^{-N+2}|A_{[m+1]}A_{[m]}^{-1}|\cdot[|\h{N-3}~~\psi(N-2)~~\psi(N-1)|_{_{[3]}}\\
    & ~ \times
    |\h{N-3},~\dt\psi(N-2),~\c E^2\psi(N-2)|_{_{[3]}}\\
    &-|\h{N-3},~\psi(N-2),~\dt\psi(N-2)|_{_{[3]}}\cdot
    |\h{N-3},~\psi(N-1),~\c E^2\psi(N-2)|_{_{[3]}}\\
    &+|\h{N-3},~\psi(N-2),~\c E^2\psi(N-2)|_{_{[3]}}\cdot
    |\h{N-3},~\psi(N-1),~\dt\psi(N-2)|_{_{[3]}}]\\
    =&0,
\end{array}
\end{equation*}
by using Proposition \ref{P:2.2} in which $\textbf{B}=(\h {N-3})$,
$(\textbf{a}, \textbf{b}, \textbf{c}, \textbf{d})=(\psi(N-2),
~\psi(N-1), ~\dt\psi(N-2), ~\c E^2\psi(N-2))$.
\end{proof}

For the explicit forms of $\psi$ together with the transformation matrices we can take either \eqref{psi1} with $c=0$
or \eqref{psi2} with $c=0$.

%%%%%%%%%%%%%%%%%%%%%%%%%%%%%%%%%%%%%%%%%%%%%%%%%%%%%%%%%%%%%%%%%%%%%%%%%%%%%%%%%%%%%%%%%%%%%%%%%%%%%%%%%%%

\subsection{Non-autonomous H2 equation}

By the parametrization \eqref{Parame} with $c=0$, i.e., $p_n=-a_n^2,~~q_m=-b_m^2$, we first rewrite the
non-autonomous H2 \eqref{H2} into
\begin{equation}\label{H2-1}
    \mathrm{H2}\equiv (u-\h{\t u}) (\t u-\h u)+(a_n^2-b_m^2)(u+\t u+\h u+\h{\t
    u}-(a_n^2+b_m^2))=0.
\end{equation}
Then, taking the transformation
\begin{subequations}\label{H2-Cond}
\begin{equation}\label{H2-Trans}
   u_{n,m}=U_{n,m}^2-2U_{n,m}\frac{g}{f}+\frac{h+s}{f},
\end{equation}
with
\begin{equation}\label{H2-Cond-3}
U_{n,m}=\sum^{n-1}_{i=n_{_0}}a_i+\sum^{m-1}_{j=m_{_0}}b_j+\gamma,~~~~(\gamma~\text{is
an arbitrary constant}),
\end{equation}
and
\begin{equation}\label{H2-Cond-2}
s-h=\alpha f, ~~~~(\alpha~\text{is some constant}),
\end{equation}
\end{subequations}
one can bilinearize \eqref{H2-1} by
\begin{subequations}\label{Bili-H2}
\begin{align}
&\mathcal{H}_1\equiv(\h g\t f-\t g \h f)+(a_n-b_m)(\h f \t f-f\h{\t f})=0,\label{Bili-H2-1}\\
&\mathcal{H}_2\equiv(g \h{\t f}-\h{\t g}f)+(a_n+b_m)(f\h{\t f}-\h f\t f)=0,\label{Bili-H2-2}\\
&\mathcal{H}_3\equiv-(a_n+b_m)\h f\t g+a_n\h{\t f}g+b_mf\h{\t
g}+\h{\t f}h-f\h{\t h}=0,\label{Bili-H2-3}\\
&\mathcal{H}_4\equiv-(a_n-b_m) f\h{\t g}+a_n\t f\h g-b_m\h f\t g+\t f\h h-\h f \t h=0,\label{Bili-H2-4}\\
&\mathcal{H}_5\equiv b_m(\h fg-f\h g)+f\h h+\h fs-g\h
g=0,\label{Bili-H2-5}
\end{align}
\end{subequations}
where the connection is
\[
    \mathrm{H2}=\sum_{i=1}^{5}\mathcal{H}_i P_i/(f\t f\h f\h{\t f}),
\]
and
\begin{subequations}%\label{Bili-Rela-H2}
\begin{align*}
&P_1=-4(a_n+b_m)[(\t U\h{\t U}-a_n^2+b_m^2)\t f\h f-\h{\t U}\h f\t g-(a_n-b_m)f\h{\t g}],\\
&\nonumber P_2=-4[(a_n-b_m)(\t U\h{\t U}-a_n^2+b_m^2)\t f\h f+(\t
U\h{\t U}-a_n^2+b_m^2)\t f\h g-\t
U\h{\t U}\h f\t g -(a_n-b_m)\t Uf\h{\t g}],\\
&P_3=4[(a_n-b_m)U\t f\h f+\h U\t f\h g-\t U\h f\t g-\t f\h h+\h f\t h],\\
&P_4=4[(a_n+b_m)(\h Uf\h{\t f}-\h f\t g)+\t U(\h{\t f}g-f\h{\t g})],\\
&P_5=4(a_n^2-b_m^2)\t f\h{\t f},
\end{align*}
\end{subequations}
\noindent where $U=U_{n,m}$ is defined in \eqref{H2-Cond-3}. This is as same as the autonomous case\cite{HZ-PartII}.

For solutions we have
\begin{prop}\label{P:3.2-H2}
The Casoratians
\begin{equation}\label{Caso-H2}
    f=|\widehat{N-1}|_{_{[3]}},~g=|\widehat{N-2},~N|_{_{[3]}},
~h=|\widehat{N-3},~N-1,~N|_{_{[3]}},~s=|\widehat{N-2},~N+1|_{_{[3]}},
\end{equation}
solve the non-autonomous bilinear equations \eqref{Bili-H2}, where
the basic column vector $\psi$ satisfies the same conditions in
Proposition \ref{P:3.1-H1}.
\end{prop}

\begin{proof}

We skip the proof for $\mathcal{H}_1$ and $\mathcal{H}_2$ as they
are just the bilinear H1. We prove $\mathcal{H}_5$ and shifted
$\mathcal{H}_3$ and $\mathcal{H}_4$ in the following forms
\begin{subequations}\label{H2-shift}
\begin{align}
&\mathcal{\dt H}_3\equiv-(a_{n-1}+b_m)\dt{\h f}g+a_{n-1}\h f\dt
g+b_m\dt f\h g+\h f\dt h-\dt f\h
    h,\label{H2-shift-3}\\
&\dth{\mathcal{H}}_4\equiv-(a_{n-1}-b_{m-1})\dth f g+\dh
f(a_{n-1}\dt g+\dt
    h)-\dt f(b_{m-1}\dh g+\dh h).\label{H2-shift-4}
\end{align}
\end{subequations}
We still use the formulae given in appendix A with $c=0$.

For \eqref{H2-shift-3}, $g=|\widehat{N-2},~N|_{_{[3]}}$, and $\h{\dt
f}$, $\dt f$, $b_m\h g- \h h$, $\h f$ and $a_{n-1}\dt g+\dt h$ are
\eqref{2-3}, \eqref{1-2}, \eqref{4-2}, \eqref{6-19} and \eqref{4-1}
respectively. Then we have
\begin{equation*}
\begin{array}{rl}
\mathcal{\dt H}_3\equiv&-(a_{n-1}+b_m)g\h{\dt f}+\dt f(b_m\h g-\h h)+\h f(a_{n-1}\dt g+\dt h)\\
    =&-a_{n-1}^{-N+2}~b_m^{-N+2}|A_{[m+1]}A_{[m]}^{-1}|\cdot[|\h{N-3},\psi(N-2),\psi(N)|_{_{[3]}}\\
    & ~ \times |\h{N-3},\dt\psi(N-2),\c E^2\psi(N-2)|_{_{[3]}}\\
    &-|\h{N-3},~\psi(N-2),~\dt\psi(N-2)|_{_{[3]}}\cdot
    |\h{N-3},~\psi(N),~\c E^2\psi(N-2)|_{_{[3]}}\\
    &+|\h{N-3},~\psi(N-2),~\c E^2\psi(N-2)|_{_{[3]}}\cdot
    |\h{N-3},~\psi(N),~\dt\psi(N-2)|_{_{[3]}}]\\
    =&0,
\end{array}
\end{equation*}
where we have made use of Proposition \ref{P:2.2} in which
$\textbf{B}=(\h {N-3})$,~ $(\textbf{a}, ~\textbf{b}, ~\textbf{c},
~\textbf{d})=(\psi(N-2), ~\psi(N), ~\dt\psi(N-2), ~\c
E^2\psi(N-2))$.

For \eqref{H2-shift-4}, $g=|\widehat{N-2},~N|_{_{[3]}}$, and $\dth
f$, $\dh f$, $a_{n-1}\dt g+\dt h$, $\dt f$ and $b_{m-1}\dh g+\dh h$
are \eqref{1-1}, \eqref{6-2},  \eqref{4-1}, \eqref{1-2} and
\eqref{6-25} respectively. Then we have
\begin{equation*}
\begin{array}{rl}
\mathcal{\dt H}_4\equiv&-(a_{n-1}-b_{m-1})\dth f g+\dh f(a_{n-1}\dt g+\dt h)-\dt f(b_{m-1}\dh g+\dh h)\\
    =&-a_{n-1}^{-N+2}~b_{m-1}^{-N+2}[|\h{N-3},~\psi(N-2),~\psi(N)|_{_{[3]}}\cdot
    |\h{N-3},~\dh\psi(N-2),~\dt\psi(N-2)|_{_{[3]}}\\
    &-|\h{N-3},~\psi(N-2),~\dh\psi(N-2)|_{_{[3]}}\cdot
    |\h{N-3},~\psi(N),~\dt\psi(N-2)|_{_{[3]}}\\
    &+|\h{N-3},~\psi(N-2),~\dt\psi(N-2)|_{_{[3]}}\cdot
    |\h{N-3},~\psi(N),~\dh\psi(N-2)|_{_{[3]}}]\\
    =&0,
\end{array}
\end{equation*}
where we have made use of Proposition \ref{P:2.2} in which
$\textbf{B}=(\h {N-3})$,~ $(\textbf{a}, ~\textbf{b}, ~\textbf{c},
~\textbf{d})=(\psi(N-2), ~\psi(N), ~\dh\psi(N-2), ~\dt\psi(N-2))$.

For \eqref{Bili-H2-5}, we have $f=|\h{N-1}|_{_{[3]}},~
g=|\h{N-2},N|_{_{[3]}},~s=|\h{N-3},N-1,N|_{_{[3]}}$, and $\h h-b_m\h
g$, $\h g - b_m \h f$ and $\h f$ are provided by \eqref{4-2},
\eqref{3-2} and \eqref{6-19}. Now we obtain
\begin{equation*}
\begin{array}{rl}
\mathcal{H}_5\equiv &f(\h h-b_m \h g)-g(\h g-b_m\h f)+\h f s\\
    =&b_m^{-N+2}
|A_{[m+1]}A_{[m]}^{-1}|\cdot[|\h{N-3},\psi(N-2),\psi(N-1)|_{_{[3]}}\cdot
    |\h{N-3},\psi(N),\c E^2\psi(N-2)|_{_{[3]}}\\
    &-|\h{N-3},~\psi(N-2),~\psi(N)|_{_{[3]}}\cdot
    |\h{N-3},~\psi(N-1),~\c E^2\psi(N-2)|_{_{[3]}}\\
    &+|\h{N-3},~\psi(N-2),~\c E^2\psi(N-2)|_{_{[3]}}\cdot
    |\h{N-3},~\psi(N-1),~\psi(N)|_{_{[3]}}]\\
    =&0,
\end{array}
\end{equation*}
where we have made use of Proposition \ref{P:2.2} in which
$\textbf{B}=(\h {N-3})$,~ $(\textbf{a}, ~\textbf{b}, ~\textbf{c},
~\textbf{d})=(\psi(N-2), ~\psi(N-1), ~\psi(N), ~\c E^2\psi(N-2))$.
\end{proof}

For the explicit forms of $\psi$ together with the transformation matrices we can take either \eqref{psi1} with $c=0$
or \eqref{psi2} with $c=0$.

%%%%%%%%%%%%%%%%%%%%%%%%%%%%%%%%%%%%%%%%%%%%%%%%%%%%%%%%%%%%%%%%%%%%%%%%%%%%%%%%%%%%%%%%%%%%%%%%%%%%%%%%%%%%%%%%%%%%%%%%%%%%

\subsection{Non-autonomous H3  equation}

With the parametrization
\begin{equation*}%\label{H3-P}
p_n=\frac{1+\alpha_{n}^2}{2\alpha_n},~~q_m=\frac{1+\beta_m^2}{2\beta_m},
~~\alpha_n^2=-\frac{a_n-c}{a_n+c},~~\beta_m^2=-\frac{b_m-c}{b_m+c},
\end{equation*}
the non-autonomous H3 equation \eqref{H3} admits two different sets
of bilinear forms. One is
\begin{subequations}\label{Bili-H3-1}
\begin{align}
&\mathcal{B}_1\equiv2c f \t f+(a_n-c)\t {\b f} \db f-(a_n+c)\b f \t {\db f}=0,\label{Bili-H3-1-1}\\
&\mathcal{B}_2\equiv2c f \h f+(b_m-c)\h {\b f} \db f-(b_m+c)\b f
\h{\db f}=0,\label{Bili-H3-1-2}
\end{align}
\end{subequations}
and the other is
\begin{subequations}\label{Bili-H3-2}
\begin{align}
&\mathcal{B}_1'\equiv(b_m+c)\b f\h{\t f}+(a_n-c)\h{\t {\b f}}f-(a_n+b_m)\t f\h {\b f}=0,\label{Bili-H3-2-1}\\
&\mathcal{B}_2'\equiv(a_n+c)f\db{\th f}+(b_m-c)\th f\db f-(a_n+b_m)\t f\h {\db f}=0,\label{Bili-H3-2-2}\\
&\mathcal{B}_3'\equiv(a_n-c)(b_m+c)\h {\db f}\t {\b
f}-(b_m-c)(a_n+c)\t {\db f}\h {\b f}-2c(a_n-b_m)f\h{\t
f}=0.\label{Bili-H3-2-3}
\end{align}
\end{subequations}
Both of them share same transformation
\begin{subequations}\label{H3-Trans}
\begin{equation}
    u_{n,m}=A~ V_{n,m}\frac{\b f_{n,m}}{f_{n,m}}+B~ V_{n,m}^{-1} \frac{\db
    f_{n,m}}{f_{n,m}},~~AB=-\frac{1}{4}\delta,
\end{equation}
where
\begin{equation}\label{H3-Cond}
V_{n,m}=\prod^{n-1}_{i=n_0}\alpha_i\prod^{m-1}_{j=m_0}\beta_j.
\end{equation}
\end{subequations}
The connections are respectively
\begin{align*}
    \mathrm{H3}=\frac{-\delta^2 B^{-2}V_{n,m}^2(a_n-c)(b_m-c)P_1+
    4\delta P_2+16B^{2}V_{n,m}^{-2}(a_n+c)(b_m+c)P_3}{32(a_n^2-c^2)(b_m^2-c^2)f\t f\h f\h{\t f}},
\end{align*}
with
\begin{align*}
&P_1=\h{\t{\b f}}[(b_m-c)\h{\b f}\mathcal{B}_1-(a_n-c)\t {\b
      f}\mathcal{B}_2]-\b f[(b_m+c)\t{\b f}\h {\mathcal{B}_1}-(a_n+c)\h {\b f}\t{\mathcal{B}_2}],\\
&P_2=2c[(b_m+c)(b_m-c)(\h f \h{\t f}\mathcal{B}_1+f \t f\h
{\mathcal{B}_1})
     -(a_n+c)(a_n-c)(\t f \h{\t f}\mathcal{B}_2+f \h f\t {\mathcal{B}_2})],\\
&P_3=\h{\t{\db f}}[(b_m+c)\h{\db f}\mathcal{B}_1-(a_n+c)\t {\db
     f}\mathcal{B}_2]-\db f[(b_m-c)\t{\db f}\h
     {\mathcal{B}_1}-(a_n-c)\h {\db f}\t{\mathcal{B}_2}],
\end{align*}
and
\begin{subequations}\label{H3-Iden2}
\begin{align}
    &\mathrm{H}3\equiv\frac{c}{f\t f\h f\h{\t f}}[A^2V_{n,m}^2\frac{\t {\b f}\h f\mathcal{B}_1'-\h {\b f}\t f \mathcal{B}_2'}{(a_n+c)(b_m+c)}
    +B^2V_{n,m}^{-2}\frac{\h {\db f}\t f \underline{\mathcal{B}}_1'-\t {\db f}\h f\underline{\mathcal{B}}_2'}{(a_n-c)(b_m-c)}~~~~~~~~~~\nonumber\\
    &~~~~~~~~+AB(\frac{\h {\db f} \t f\mathcal{B}_2'+\t {\b f}\h f\underline{\mathcal{B}}_2'}{(a_n+c)(b_m-c)}-
    \frac{\t {\db f} \h f\mathcal{B}_1'+\h {\b f}\t
    f\underline{\mathcal{B}}_1'}{(a_n-c)(b_m+c)}+
    \frac{2(a_n+b_m)\t f\h f\mathcal{B}_3'}{(a_n^2-c^2)(b_m^2-c^2)})],\nonumber
\end{align}
\end{subequations}
which are similar to the one  in the autonomous case\cite{HZ-PartII}. For
solutions we have
\begin{prop}\label{P:3.3-H3}
The Casoratians
\begin{equation}\label{C-H3}
    f=|\widehat{N-1}|_{_{[\nu]}},~~\nu=1,~2,~3,
\end{equation}
solve non-autonomous bilinear equations \eqref{Bili-H3-1} and
\eqref{Bili-H3-2}, if the basic column vector $\psi(n,m,l)$ is
symmetric in terms of pairs $(n, a_n)$ and $(m,b_m)$, and together
with auxiliary vectors $\omega(n,m,l)$ and $\zeta(n,m,l)$, satisfies
the following shift relations
\begin{subequations}\label{H3-1-ditui}
\begin{align}
&(c-a_n)\db\psi=\psi-\t{\db\psi},\label{H3-1-ditui-1}\\
&\psi=A_{[n]}\omega,~~(a_n+c)\t\omega=\omega+\bt\omega,\label{H3-1-ditui-2}\\
&\psi=B_{[l]}\zeta,~~(c+b_m)\b \zeta=
\zeta+\bh\zeta,\label{H3-1-ditui-3}
\end{align}
\end{subequations}
\noindent where $A_{[n]}$ and $B_{[l]}$ are $N\times N$ transform
matrices, and the matrix product $B_{[l]}B_{[l+1]}^{-1}$ is
independent of $l$.
\end{prop}

\begin{proof}

We prove \eqref{Bili-H3-1-1} in its down-tilde version

\begin{equation}\label{dt-Bili-H3-1-1}
\mathcal{\dt B}_1\equiv2c\dt f f+(a_{n-1}-c)\b f\dtb
f-(a_{n-1}+c)\b{\dt f}\db f.
\end{equation}
In \eqref{dt-Bili-H3-1-1}, for $f$, $\dt f$, $\b f$, $\dtb f$,
$\dt{\b f}$ and $\db f$, we make use of \eqref{1-10}, \eqref{6-8},
\eqref{1-9}, \eqref{2-6}, \eqref{2-7} and \eqref{1-6} respectively,
and get
\begin{align}
    \mathcal{\dt
    B}_1\equiv&\prod_{j=0}^{N-3}(a_{n-1}-b_{m+j})^{-1}\prod_{j=0}^{N-3}(c-b_{m+j})^{-1}\prod_{j=0}^{N-3}(c+b_{m+j})^{-1}
    |B_{[l+1]}B_{[l]}^{-1}|\nonumber\\
    &\times[-|\h{N-3},~\psi(N-2),~\dt\psi(N-2)|_{_{[2]}}\cdot|\h{N-3},~\db\psi(N-2),~\c
    E^3\psi(N-2)|_{_{[2]}}\nonumber\\
    &+|\h{N-3},~\psi(N-2),~\db\psi(N-2)|_{_{[2]}}\cdot|\h{N-3},~\dt\psi(N-2),~\c E^3\psi(N-2)|_{_{[2]}}\nonumber\\
    &-|\h{N-3},~\psi(N-2),~\c E^3\psi(N-2)|_{_{[2]}}\cdot|\h{N-3},~\dt\psi(N-2),~\db\psi(N-2)|_{_{[2]}}]\nonumber\\
    =&0,\nonumber
\end{align}
with the help of Proposition \ref{P:2.2} in which
$\textbf{B}=(\h{N-3})$, $(\textbf{a}, \textbf{b}, \textbf{c},
\textbf{d})=(\psi(N-2), \dt\psi(N-2), \db\psi(N-2), \c
E^3\psi(N-2))$. Here to get the coefficient
$|B_{[l+1]}B_{[l]}^{-1}|$ we request $B_{[l]}B_{[l+1]}^{-1}$ to be
independent of $l$, i.e.
$B_{[l]}B_{[l-1]}^{-1}=B_{[l+1]}B_{[l]}^{-1}$.

$\mathcal{B}_2$ holds thanks to $\mathcal{B}_1$ and the \emph{n-m}
symmetric property of $\psi(n,m,l)$.

We prove $\mathcal{B}_1'$ in its down-tilde shifted version, i.e.
\begin{equation}\label{dt-Bili-H3-2-1}
\mathcal{\dt B}_1'\equiv(b_m+c)\dt{\b f}\h f+(a_n-c)\hb f\dt
f-(a_{n-1}+b_m)\hb{\dt f}f.
\end{equation}
In \eqref{dt-Bili-H3-2-1}, we take $f=|\widehat{N-1}|_{_{[3]}}$, and
$\dt{\b f}$, $\h f$, $\hb f$, $\dt f$ and $\hb{\dt f}$ as
\eqref{6-1}, \eqref{1-4}, \eqref{6-5}, \eqref{1-1} and \eqref{6-24},
respectively. Then we have
\begin{subequations}
\begin{align}
    \mathcal{\dt
    B}_1'\equiv&(a_{n-1}-c)^{-N+2}(b_m+c)^{-N+2}|A_{[m+1]}A_{[m]}^{-1}|\nonumber\\
    &\times[-|\psi(0),~\t{N-2},~\psi(N-1)|_{_{[3]}}\cdot|\t{N-2},~\dt\psi(N-1),~\c
    E^2\psi(N-1)|_{_{[3]}}\nonumber\\
    &+|\psi(0),~\t{N-2},~\dt\psi(N-1)|_{_{[3]}}\cdot|\t{N-2},~\psi(N-1),~\c E^2\psi(N-1)|_{_{[3]}}\nonumber\\
    &-|\psi(0),~\t{N-2},~\c E^2\psi(N-1)|_{_{[3]}}\cdot|\t{N-2},~\psi(N-1),~\dt\psi(N-1)|_{_{[3]}}]\nonumber\\
    =&0,\nonumber
\end{align}
\end{subequations}
where we have made use of Proposition \ref{P:2.2}, in which
$\textbf{B}=(\t{N-2})$, $(\textbf{a}, \textbf{b}, \textbf{c},
\textbf{d})=(\psi(0), \psi(N-1), \dt\psi(N-1), \c E^2\psi(N-1))$.

$\mathcal{B}_2'$ holds thanks to $\mathcal{B}_1'$ and the
\emph{n-m} symmetric property of $\psi(n,m,l)$.

By a down-hat shift, $\mathcal{B}_3'$ reads
\begin{equation}\label{dh-Bilinear-H3-2-3}
\mathcal{\dh B}_3'\equiv(a_n-c)(b_{m-1}+c)\db f\dh{\tb
f}-(a_n+c)(b_{m-1}-c)\b f\t{\dhb f}-2c(a_n-b_{m-1})\dh f\t f.
\end{equation}
In \eqref{dh-Bilinear-H3-2-3}, for $\db f$, $\dh {\tb f}$, $\b f$,
$\dhb {\t f}$, $\dh f$ and $\t f$ we use \eqref{1-12}, \eqref{6-6},
\eqref{1-8}, \eqref{2-8}, \eqref{1-7} and \eqref{6-7} respectively.
Then we have
\begin{subequations}
\begin{align}
    \mathcal{\dh
    B}_3'\equiv&\prod_{i=1}^{N-2}(b_{m-1}-a_{n+i})^{-1}\prod_{i=1}^{N-2}(c-a_{n+i})^{-1}\prod_{i=1}^{N-2}(c+a_{n+i})^{-1}
   |B_{[l+1]}B_{[l]}^{-1}|\nonumber\\
    &\times[|\t{N-2},~\psi(0),~\db\psi(N-1)|_{_{[1]}}\cdot|\t{N-2},~\dh\psi(N-2),~\c
    E^3\psi(N-1)|_{_{[1]}}\nonumber\\
    &-|\t{N-2},~\psi(0),~\dh\psi(N-1)|_{_{[1]}}\cdot|\t{N-2},~\db\psi(N-1),~\c E^3\psi(N-1)|_{_{[1]}}\nonumber\\
    &+|\t{N-2},~\psi(0),~\c E^3\psi(N-1)|_{_{[1]}}\cdot|\t{N-2},~\db\psi(N-1),~\dh\psi(N-1)|_{_{[1]}}]\nonumber\\
    =&0,\nonumber
\end{align}
\end{subequations}
where we have made use of Proposition \ref{P:2.2} in which
$\textbf{B}=(\t{N-2})$, $(\textbf{a}, \textbf{b}, \textbf{c},
\textbf{d})=(\psi(0), \db\psi(N-1), \dh\psi(N-1), \c E^3\psi(N-1))$.
\end{proof}

For the explicit forms of $\psi$ together with the transformation matrices we can take either \eqref{psi1}
or \eqref{psi2}.

%%%%%%%%%%%%%%%%%%%%%%%%%%%%%%%%%%%%%%%%%%%%%%%%%%%%%%%%%%%%%%%%%%%%%%%%%%%%%%%%%%%%%%%%%%%%%%%%%%%%%%%%%%%%%%%%%%%%%%%%%%%%

\subsection{Non-autonomous Q1 equation}

The non-autonomous Q1 equation can have two different
bilinearizations  which are also similar to their autonomous
cases\cite{HZ-PartII}. First, using the parametrization
\begin{subequations}\label{Q1Trans}
\begin{equation}\label{Tran1}
    p_n=\frac{rc^2}{a_n^2-c^2},~~~~q_m=\frac{rc^2}{b_m^2-c^2},
\end{equation}
where $c$ and $r$ are constants,
and through the transformation
\begin{equation}\label{Tran2}
    u_{n,m}=AV_{n,m}\frac{\b{\b{f}}_{n,m}}{f_{n,m}}+BV_{n,m}^{-1}\frac{\db{\db{f}}_{n,m}}{f_{n,m}},~~AB=\frac{r^2\delta^2}{16},
\end{equation}
where $V_{n,m}$ is \eqref{H3-Cond} and
$\alpha_n=\frac{a_n-c}{a_n+c},~~\beta_m=\frac{b_m-c}{b_m+c}$,
\end{subequations}
the non-autonomous Q1 equation \eqref{Q1} can be transformed to
\eqref{Bili-H3-1}, i.e., one of bilinear form for the non-autonomous
H3 equation. So its solutions consequently follow the Proposition
\ref{P:3.3-H3}. In this case, the connection is
\begin{equation*}\label{Q1-Iden1}
    \mathrm{Q}1\equiv
    (\alpha_n\beta_mU_{n,m}^2A^2\overline{P}_1
    +\frac{r^2}{16}\alpha_n^{-1}\beta_m^{-1}(a_n+c)^{-2}(b_m+c)^{-2}\delta^2P_2
    +\alpha_n^{-1}\beta_m^{-1}U_{n,m}^{-2}B^2\underline{P}_1)/f\t f\h f\h {\t
    f},
\end{equation*}
where
\begin{equation*}
\begin{array}{rl}
  P_1=& Y\t Y-X\h X,~~~~ X=\mathcal{B}_1-2cf\t f,~~ ~~Y=\mathcal{B}_2-2cf\h f,\\
  P_2=&(a_n^2-c^2)(b_m+c)^2 \Bigl( \b{X}\db{{\h X}}- 4c^2 \b f\,\tb f \db {\h f}\db {\th f} \,\Bigr)
       +(a_n^2-c^2)(b_m-c)^2 \Bigl(\db X \bh X -4c^2 \db f\,\db{\t f}\,\hb f\,\th {\b f}\,\Bigr)\\
      &-4c^2(b_m^2-c^2)\Bigl(X \h X -4c^2 f\t f \,\h f\, \th f \, \Bigr)
      -(b_m^2-c^2)(a_n+c)^2 \Bigl(\b Y \db{\t Y}-4c^2 \b f\,\bh f \,\db {\t f}\,\db {\th f}\,
      \Bigr)\\
      &-(b_m^2-c^2)(a_n-c)^2 \Bigl(\db Y \bt Y - 4c^2 \db f\db{\h f}\,\tb f\,\th {\b f}\, \Bigr)
        +4c^2(a_n^2-c^2) \Bigl(Y \t Y-4c^2 f\h f \,\t f\, \th f\, \Bigr).
\end{array}
\end{equation*}

The second bilinear form employs the transformation
\begin{subequations}\label{Q1Tran2}
\begin{equation}
    u_{n,m}=W_{n,m}-(\frac{c^2}{r}-\delta^2r)\frac{g_{n,m}}{f_{n,m}},
\end{equation}
where $c,r$ are constants,
\begin{equation}
W_{n,m}=\sum_{i=n_0}^{n-1}\alpha_i+\sum_{j=m_0}^{m-1}\beta_j,~\alpha_n=p_na_n,~\beta_m=q_mb_m,
~p_n=\frac{c^2/r-\delta^2r}{a_n^2-\delta^2},
~q_m=\frac{c^2/r-\delta^2r}{b_m^2-\delta^2},
\end{equation}
\end{subequations}
and the bilinear form reads
\begin{subequations}\label{Bili-Q1-2}
\begin{eqnarray}
\mathcal Q_1&\equiv&(b_m-\delta){\h{\tb f}}f+(a_n+\delta )
\h{\t f} \b{f}-(a_n+b_m){\tb{f}}\h{f}=0,\label{Bili-Q1-2-1}\\
\mathcal Q_2&\equiv &  (a_n-b_m){\h{\tb f}}f
+(b_m+\delta){\tb{f}}\h{f}-(a_n+\delta )\t{f}{\hb{f}}=0,\label{Bili-Q1-2-2}\\
\mathcal Q_3&\equiv &\t{f}{\hb{f}}-{\tb{f}} \h{f}+(b_m-\delta)
{\hb{f}} \t{g}-(a_n-\delta ){\tb{f}} \h{g}
+ (a_n-b_m)\b{f} \h{\t{g}}=0,\label{Bili-Q1-2-3}\\
\mathcal Q_4&\equiv&(a_n-b_m)({\h{\tb f}} g -\b{f} \h{\t{g}})+
(a_n+b_m)({\tb{f }} \h{g} -{\hb{f }}\t{g}) =0.\label{Bili-Q1-2-4}
\end{eqnarray}
\end{subequations}
The connection is
\begin{equation*}
\mathrm{Q}1 =\frac{(c^2/r-\delta^2r)^3}
{(a_n^2-\delta^2)(b_m^2-\delta^2)(a_n-b_m)(a_n+\delta)\b{f} f \t f
\h f \h{\t f}}\quad \sum^{4}_{i=1}\mathcal{Q}_iP_i,
\end{equation*}
where
\begin{equation*}
\begin{array}{rl}
P_1~=&(a_n-b_m)~[-(a_n-b_m)\t{f}\h{f}g
+(a_n+b_m)f(\h{f} \t{g}-\t{f}\h{g})\\
&\phantom{(a_n-b_m)}-(a_n^2-\delta ^2)\t{f}\h{g}g +(b_m^2-\delta
^2)\h{f}\t{g}g+(a_n^2-b_m^2)f\t{g}\h{g}~],\\
P_2~=&(a_n+b_m)~[(a_n-b_m)\t{f}\h{f}g+(b_m-\delta
)f(\t{f}\h{g}-\h{f}\t{g})
+(a_n-b_m)(b_m-\delta )\t{g}(\h{f}g-f\h{g})],\\
P_3~=&(a_n+b_m)(a_n+\delta)~[(a_n-b_m)\t{f}\h{f}g+(b_m-\delta)\t{f}f\h{g}
-(a_n-\delta)\h{f}f\t{g}~],\\
P_4~=&(a_n+\delta)f[-(a_n-b_m)\t{f}\h{f}+(a_n-\delta)(b_m-\delta
)(\t{f}\h{g}-\h{f}\t{g})].
\end{array}
\end{equation*}
For solutions to \eqref{Bili-Q1-2} we have
\begin{prop}\label{P:3.4-Q1}
The Casoratians
\begin{equation}\label{Q1Caso}
f=|\widehat{N-1}|_{_{[3]}},~g=|-1,\widetilde{N-1}|_{_{[3]}},
\end{equation}
solve the non-autonomous bilinear equations \eqref{Bili-Q1-2}, if
their basic column vector $\psi(n,m,l)$ has symmetric property and
satisfies the shift relations \eqref{H3-1-ditui} with $c=\delta$, as
well as
\begin{equation}
\psi=A_{[n]}A_{[m]}\sigma,~(a_n+\delta)\t\sigma=\sigma+\t{\b\sigma},~(b_m+\delta)\h\sigma=\sigma+\h{\b\sigma},\label{Q1-Recu-2}
\end{equation}
where the matrices $A_{[n]}$, $A_{[m]}$ and their shifts are in an Abelian group.
\end{prop}

\begin{proof}

The down-hat-bar shifted version of \eqref{Bili-Q1-2-1} is
\begin{equation}\label{dbh-Bili-Q1-2-1}
   \mathcal{\dhb Q}_1\equiv(b_{m-1}-\delta)\t f\dhb f+(a_n+\delta)
\t{\db f}\dh f-(a_n+b_{m-1})\t{\dh f}\db f,
\end{equation}
and by using \eqref{1-3}, \eqref{6-3}, \eqref{6-4},  \eqref{6-2},
\eqref{2-2} and $\db f=|-1~~\h{N-3}~~\psi(N-2)|_{[3]}$, with
$c=\delta$ we can prove it is true. In fact, it is the same as the
down hat-bar version of $\mathcal{B}_2'$.

By a down-tilde shift \eqref{Bili-Q1-2-2} is written as
\begin{equation}\label{dt-Bili-Q1-2-2}
   \mathcal{\dt Q}_2\equiv(a_{n-1}-b_m){\h{\b f}}\dt f
+(b_m+\delta){\b{f}}\h{\dt f}-(a_{n-1}+\delta )f\hb{\dt f}.
\end{equation}
Thanks to Proposition \ref{P:2.1}, here we use
$f=|\widehat{N-1}|_{_{[2]}}$. Then with the help of \eqref{6-21},
\eqref{1-5}, \eqref{1-11}, \eqref{6-23} and \eqref{6-22}, with
$c=\delta$, we can rewrite \eqref{dt-Bili-Q1-2-2} as
\begin{equation*}\label{l-dt-Bili-Q1-2-2}
\begin{array}{rl}
   \mathcal{\dt Q}_2~\equiv&\prod_{j=1}^{N-2}(\delta+b_{m+j})^{-1}\prod_{j=1}^{N-2}(a_{n-1}-b_{m+j})^{-1}
   |B_{[l+1]}B_{[l]}^{-1}|~[|\widetilde{N-1},~\c
    E^3\psi(N-1)|_{_{[2]}}\\
    & \times |\widehat{N-2},~\dt\psi(N-1)|_{_{[2]}}-|\widehat{N-2},~\c
    E^3\psi(N-1)|_{_{[2]}}\cdot|\widetilde{N-1},~\dt\psi(N-1)|_{_{[2]}}\\
    &-|\widehat{N-1}|_{_{[2]}}\cdot|\widetilde{N-2},~\dt\psi(N-1),~\c
    E^3\psi(N-1)|_{_{[2]}}]\\
    &=0,
   \end{array}
\end{equation*}
where we have made use of Proposition \ref{P:2.2}.

By a down-tilde-hat shift \eqref{Bili-Q1-2-3} is written as
\begin{equation}\label{dth-Bili-Q1-2-3}
   \mathcal{\dth Q}_3\equiv-\dh {\b f}[\dt f+(a_{n-1}-\delta)\dt g]
   +\dt {\b f}[\dh f+(b_{m-1}-\delta)\dh g]+(a_{n-1}-b_{m-1})\dh{\dt{\b
   f}}g.
\end{equation}
 With $c=\delta$, substituting \eqref{6-10},\eqref{5-1}, \eqref{5-3}, \eqref{5-2},
\eqref{6-9}, \eqref{6-11} and \eqref{5-4} into the right-hand side
of \eqref{dth-Bili-Q1-2-3}, we obtain
\begin{equation*}
   \mathcal{\dt{ \dh
   Q}}_3~\equiv(a_{n-1}-\delta)^2Y_1-(b_{m-1}-\delta)^2Y_2+(a_{n-1}-\delta)^2(b_{m-1}-\delta)^2Y_3,
\end{equation*}
where
\begin{subequations}%\label{Y}
\begin{align}%\label{Ymu}
Y_{\mu}=&f|E_{\mu}{\psi}(-1),\psi(-1),\t{N-2}|_{_{[3]}}+\db{f}|E_{\mu}{\psi}(-1),\t{N-1}|_{_{[3]}}
-g|E_{\mu}{\psi}(-1),\h{N-2}|_{_{[3]}},
~~\mu=1,~2,\nonumber\\
Y_3=&|\dt{\psi}(-1),\psi(-1),\t{N-2}|_{_{[3]}}|\dh{\psi}(-1),\t{N-1}|_{_{[3]}}
-|\dh{\psi}(-1),\psi(-1),\t{N-2}|_{_{[3]}}|\dt{\psi}(-1),\t{N-1}|_{_{[3]}}\nonumber\\
&+g|\dh{\psi}(-1),\dt{\psi}(-1),\t{N-2}|_{_{[3]}},\nonumber%\label{Y3}
\end{align}
\end{subequations}
which are zeros in the light of  Proposition \ref{P:2.2} \cite{HZ-PartII}.
Thus, we have proved $\mathcal Q_3=0$.  To prove $\mathcal
Q_4=0$, we go to prove ${\mathcal Q}_4^{'}=\mathcal Q_3+\mathcal Q_4=0$,
i.e.,
\begin{equation}\label{plus-Bili-Q1-2-4}
   \mathcal{Q}_4^{'}\equiv-\h {\b f}[\t f-(a_{n}+\delta)\t g]
  -\t {\b f}[\h f-(b_{m}+\delta)\h g]+(a_{n}-b_{m})\h{\t{\b f}}g=0,
\end{equation}
In the light of \eqref{6-13}, \eqref{5-5}, \eqref{5-7}, \eqref{5-6},
\eqref{6-12}, \eqref{6-14} and \eqref{5-8}, with $c=\delta$ we can
rewrite $\mathcal{Q}_4^{'}$ as
\begin{equation*}
   \mathcal{
   Q}_4^{'}~\equiv(a_{n}+\delta)^2Z_1-(b_{m}+\delta)^2Z_2-(a_{n}+\delta)^2(b_{m}+\delta)^2Z_3,
\end{equation*}
where
\begin{subequations}%\label{Z}
\begin{align}
Z_{\mu}=&f|\c E^{\mu}{\psi}(-1),\psi(-1),\t{N-2}|_{_{[3]}}+\db{f}|\c
E^{\mu}{\psi}(-1),\t{N-1}|_{_{[3]}} -g|\c
E^{\mu}{\psi}(-1),\h{N-2}|_{_{[3]}},~~\mu=1, 2,\nonumber\\%\label{Zmu}
Z_3=&|\c E^{1}{\psi}(-1),\psi(-1),\t{N-2}|_{_{[3]}}|\c
E^{2}{\psi}(-1),\t{N-1}|_{_{[3]}}-|\c
E^{2}{\psi}(-1),\psi(-1),\t{N-2}|_{_{[3]}}\nonumber\\
&\times|\c E^{1}{\psi}(-1),\t{N-1}|_{_{[3]}}+g|\c E^{2}{\psi}(-1),\c
E^{1}{\psi}(-1),\t{N-2}|_{_{[3]}}, \nonumber%\label{Z3}
\end{align}
\end{subequations}
which are also zeros\cite{HZ-PartII} in the light of Proposition \ref{P:2.2}.
Thus we  finish the proof.
\end{proof}

For the explicit forms of $\psi$ together with the transformation matrices we can take either \eqref{psi1} with $c=\delta$
or \eqref{psi2} with $c=\delta$.
%%%%%%%%%%%%%%%%%%%%%%%%%%%%%%%%%%%%%%%%%%%%%%%%%%%%%%%%%%%%%%%%%%%%%%%%%%%%%%%%%%%%%%%%%%%%%%%%%%%%

\section{Conclusions}

We have derived bilinear forms and Casoratian solutions for the non-autonomous H1, H2, H3$_\delta$ and Q1$_\delta$ models in the
non-autonomous ABS list.
The transformations that we used to fulfill bilinearization are quite similar to those used in the
autonomous cases\cite{HZ-PartII}.
Besides, the bilinear forms and Casoratian structures of solutions are also similar to the autonomous cases.
In addition, in Appendix we listed Caosratian shift formulae for non-autonomous case.

It would be interesting to have a look at the non-autonomous deformation
in terms of the bilinearizations, bilinear equations and Casoratian vectors.
With comparison we can sum up the following deformations from autonomous case to non-autonomous case:
\begin{equation*}
\begin{array}{llcl}
\mathrm{spacing~ paramaters:}& (a,b) & \to& (a_n, b_m),\\
\mathrm{linear~ function:}& an+bm & \to& \sum^{n-1}_{i=n_0}a_i+\sum^{m-1}_{j=m_0}b_j,\\
\mathrm{discre~ exponential~ function:}&
\left(\frac{a+k}{a-k}\right)^n\left(\frac{b+k}{b-k}\right)^m & \to&
\prod_{i=n_0}^{n-1}(\frac{a_i+k}{a_i-k})\prod_{j=m_0}^{m-1}(\frac{b_j+k}{b_j-k})
\end{array}
\end{equation*}
We note that many discrete bilinear equations can be deautonomised
using these deformations (see \cite{Kajiwara-Ohta-1}). Besides,
actually, here we have seen  that these deformations can well keep
the correspondence of autonomous and non-autonomous ABS lattice
equations.

%%%%%%%%%%%%%%%%%%%%%%%%%%%%%%%%%%%%%%%%%%%%%%%%%%%%%%%%%%%%%%%%%%%%%%%%%%%%%%%%%%%%%%%%%%%%%%%%%%%%%%%%%%%%%%%%%%%%%%%
\subsection*{Acknowledgements}
This project is supported by the NSF of China (11071157), Shanghai
Leading Academic Discipline Project (No.J50101) and Postgraduate
Innovation Foundation of Shanghai University (No. SHUCX111027).

%%%%%%%%%%%%%%%%%%%%%%%%%%%%%%%%%%%%%%%%%%%%%%%%%%%%%%%%%%%%%%%%%%%%%%%%%%%%%%%%%%%%%%%%%%%%%%%%%%%%%%%%%%%%%%%%%%%%%%%
\begin{appendix}
\section{Casoratian shift formulae}
\label{A:a}

We derive shift formulae for the Casoratians
\begin{equation}
    f=|\h {N-1}|_{_{[\nu]}},~g=|\h {N-2}, N|_{_{[\nu]}},~h=|\h {N-2},~s=|\h {N-3}, N-1, N|_{_{[\nu]}},
    N+1|_{_{[\nu]}},
\end{equation}
for $\nu=1, 2, 3,$,
where the basic column vector $\psi(n,m,l)$ satisfies the following relation
\begin{subequations}\label{Appendix-1}
\begin{align}\label{Appendix-1-1}
(a_{n-1}-c)\dt\psi=\psi-\dt{\b
\psi},~~(b_{m-1}-c)\dh\psi=\psi-\dh{\b \psi},
\end{align}
and its auxiliary vectors $\omega(n,m,l)$, $\phi(n,m,l)$,
$\zeta(n,m,l)$ and $\sigma(n,m,l)$ satisfy
\begin{align}
&\psi=A_{[n]}\omega,~~(a_n+c)\t\omega=\omega+\t{\b
\omega},\label{Appendix-1-2}\\
&\psi=A_{[m]}\phi,~~(a_n+c)\h\phi=\phi+\h{\b
\phi},\label{Appendix-1-3}\\
&\psi=B_{[l]}\zeta,~~(c+b_m)\b\zeta=\zeta+\h{\b \zeta},\label{Appendix-1-4}\\
&\psi=A_{[n]}A_{[m]}\sigma,~~(a_n+c)\t\sigma=\sigma+\t{\b\sigma},~~(b_m+c)\h\sigma=\sigma+\h{\b\sigma},\label{Appendix-1-5}
\end{align}
\end{subequations}
where $c$ is a constant, $A_{[n]}$,~$A_{[m]}$ and $B_{[l]}$ only depends on $n$,
$m$ and $l$, respectively, $A_{[n]}, A_{[m]}$ and their shifts are in an Abelian group, and the matrix product
$B_{[l+1]}B_{[l]}^{-1}$ is independent of $l$.

We give two explicit forms for  $\psi$ satisfying the above criterion \eqref{Appendix-1}.
One is
\begin{subequations}\label{psi1}
\begin{eqnarray}
\psi(n,m,l)&=&\psi^+(n,m,l)+\psi^-(n,m,l),\\
\psi^{\pm}(n,m,l)&=&(\psi^\pm_1(n,m,l),\psi^\pm_2(n,m,l),\cdots,\psi^\pm_{N}(n,m,l))^T,
\end{eqnarray}
with
\begin{equation}
\psi_r^\pm(n,m,l)=\rho_r^\pm(c\pm
k_r)^l\prod_{i=n_{_0}}^{n-1}(a_i\pm
k_r)\prod_{j=m_{_0}}^{m-1}(b_j\pm k_r),~~r=1,2,\cdots,N,
\end{equation}
where $\rho_r^\pm$ and $k_r$ are constants.
The available  transform matrices are
%\begin{subequations}\label{AnmBl}
\begin{eqnarray}
A_{[n]}=\mbox{Diag}(A_{{[n]}_{1}}(k_1,n),\cdots,A_{{[n]}_{N}}(k_N,n)),~~
&A_{{[n]}_{r}}(k_r,n)=\prod_{i=n_0}^{n-1}(a_i^2-k_r^2),\label{An}\\
A_{[m]}=\mbox{Diag}(A_{{[m]}_{1}}(k_1,m),\cdots,A_{{[m]}_{N}}(k_N,m)),~~
&A_{{[m]}_{r}}(k_r,m)=\prod_{j=m_0}^{m-1}(b_j^2-k_r^2),\label{Amm}\\
B_{[l]}=\mbox{Diag}(B_{{[l]}_{1}}(k_1,l),\cdots,B_{{[l]}_{N}}(k_N,l)),~~
&B_{{[l]}_{r}}(k_r,l)=(c^2-k_r^2)^l,\label{Bl}
\end{eqnarray}
\end{subequations}
and the auxiliary vectors $\omega,\phi,\zeta,\sigma$
are correspondingly defined by these transform matrices and $\psi$ through \eqref{Appendix-1-2}-\eqref{Appendix-1-5}.
Another explicit form of $\psi$ is
\begin{subequations}\label{psi2}
\begin{eqnarray}
\psi(n,m,l)&=&\mathcal{A}_+\psi^+(n,m,l)+\mathcal{A}_-\psi^-(n,m,l),\\
\psi^{\pm}(n,m,l)&=&(\psi^\pm_1(n,m,l),\psi^\pm_2(n,m,l),\cdots,\psi^\pm_{N}(n,m,l))^T,
\end{eqnarray}
with
\begin{equation}
\psi^\pm_r(n,m,l)=\frac{1}{(r-1)!}\partial_{k_1}^{r-1}[\rho_1^\pm(c\pm
k_1)^l\prod_{i=n_{_0}}^{n-1}(a_i\pm
k_1)\prod_{j=m_{_0}}^{m-1}(b_j\pm k_1)],~~r=1,2,\cdots,N,
\end{equation}
where $\mathcal{A}_{\pm}$ are two arbitrary non-singular lower
triangular Toeplitz matrix(see \cite{ZDJ-Wronskian}). The available
transform matrices of this case are
\begin{eqnarray}
A_{[n]}=(a_{s,i}(k_1))_{N\times N},~a_{s,i}(k_1)&=&\left\{
\begin{array}{ll}
 \frac{\partial_{k_1}^{s-i}}{(s-i)!}\displaystyle\prod_{i=n_0}^{n-1}(a_i^2-k_1^2), & \hbox{$s\geq i$,} \\
 0, &
\hbox{$s<i$,}
\end{array}
\right.  ~s,i=1,\cdots,N,\\
A_{[m]}=(a_{s,j}(k_1))_{N\times N},~a_{s,j}(k_1)&=&\left\{
\begin{array}{ll}
\frac{\partial_{k_1}^{s-j}}{(s-j)!}\displaystyle\prod_{j=m_0}^{m-1}(b_j^2-k_1^2), & \hbox{$s\geq j$,} \\
 0, &
\hbox{$s<j$,}
\end{array}
\right. ~s,j=1,\cdots,N,~~~~~~\\
B_{[l]}=(b_{s,j}(k_1))_{N\times N},~b_{s,j}(k_1)&=&\left\{
\begin{array}{ll}
 \frac{\partial_{k_1}^{s-j}}{(s-j)!}(c^2-k_1^2)^l, & \hbox{$s\geq j$,} \\
 0, &
\hbox{$s<j$,}
\end{array}
\right.  ~s,j=1,\cdots,N.
\end{eqnarray}
\end{subequations}
The auxiliary vectors $\omega,\phi,\zeta,\sigma$ are also
correspondingly defined by these transform matrices and $\psi$
through \eqref{Appendix-1-2}-\eqref{Appendix-1-5}. Since $A_{[n]},
A_{[m]}, B_{[l]}$ are non-singular lower triangular Toeplitz
matrices which compose an Abelian group (see \cite{ZDJ-Wronskian}),
the commutative property holds automatically and the matrix product
$B_{[l+1]}B_{[l]}^{-1}$ is independent of $l$ (see
\cite{SZ-SIGMA-2010}).

In the following we list some Casoratian shift formulae.
For convenience, we define operators $\c E^\nu$, $\nu=1, 2, 3$, as
follows
\begin{equation}\label{Appendix-3}
\c E^1\psi=A_{[n]}A_{[n+1]}^{-1}E^1\psi,~~\c
E^2\psi=A_{[m]}A_{[m+1]}^{-1}E^2\psi,~~ \c
E^3\psi=B_{[l]}B_{[l+1]}^{-1}E^3\psi.
\end{equation}
%
%Those above shift relations in \eqref{Appendix-1} and
%\eqref{Appendix-2} can produce many further Casoratian formulae, as
%follows:\\
These Casoratian shift formulae are
\begin{subequations}\label{1}
\begin{align}
&(a_{n-1}-c)^{N-1}\dt f_{_{[3]}}=|\h{N-2},~\dt \psi(N-1)|_{_{[3]}},\label{1-1}\\
&(a_{n-1}-c)^{N-2}\dt f_{_{[3]}}=-|\h{N-2},~\dt \psi(N-2)|_{_{[3]}},\label{1-2}\\
&(b_{m}+c)^{N-1}\h f_{_{[3]}}=|A_{[m+1]}A_{[m]}^{-1}|\cdot|\h{N-2},~\c E^2\psi(N-1)|_{_{[3]}},\label{1-4}\\
&(a_{n}+c)^{N-2}\t f_{_{[3]}}=|A_{[n+1]}A_{[n]}^{-1}|\cdot|\h{N-2},~\c E^1\psi(N-2)|_{_{[3]}},\label{1-3}\\
&(b_{m-1}-c)^{N-2}\dh f_{_{[3]}}=-|\h{N-2},~\dh \psi(N-2)|_{_{[3]}},\label{6-2}\\
&(b_{m-1}-c)^{N-1}\db {\dh f}_{_{[3]}}=|-1,~\h{N-3},~\dh \psi(N-2)|_{_{[3]}},\label{6-3}\\
&(b_{m}+c)^{N-2}\h f_{_{[3]}}=|A_{[m+1]}A_{[m]}^{-1}
|\cdot|\h{N-2},~\c E^2\psi(N-2)|_{_{[3]}},~\label{6-19}\\
&(b_{m}+c)^{N-2}\h {\b f}_{_{[3]}}=|A_{[m+1]}A_{[m]}^{-1}
|\cdot|\t{N-1},~\c E^2\psi(N-1)|_{_{[3]}};\label{6-5}
\end{align}
\end{subequations}
\begin{subequations}\label{2}
\begin{align}
&\prod_{j=0}^{N-2}(a_{n-1}-b_{m+j})\dt f_{_{[2]}}=|\h{N-2},~\dt \psi(N-1)|_{_{[2]}},\label{1-5}\\
&\prod_{j=0}^{N-3}(a_{n-1}-b_{m+j})\dt f_{_{[2]}}=-|\h{N-2},~\dt \psi(N-2)|_{_{[2]}},\label{1-10}\\
&\prod_{j=0}^{N-3}(c-b_{m+j})\db f_{_{[2]}}=-|\h{N-2},~\db \psi(N-2)|_{_{[2]}},\label{1-6}\\
&\prod_{i=0}^{N-2}(c-a_{n+i})\db f_{_{[1]}}=|\h{N-2},~\db \psi(N-1)|_{_{[1]}},\label{1-12}\\
&\prod_{i=0}^{N-2}(c+a_{n+i})\b f_{_{[1]}}=|B_{[l+1]}B_{[l]}^{-1}|\cdot|\h{N-2},~\c E^3 \psi(N-1)|_{_{[1]}},\label{1-8}\\
&\prod_{j=0}^{N-3}(c+a_{n+i})\b
f_{_{[1]}}=|B_{[l+1]}B_{[l]}^{-1}|\cdot|\h{N-2},~\c E^3
\psi(N-2)|_{_{[1]}},\label{1-13}\\
&\prod_{j=0}^{N-2}(c+b_{m+j})\b
f_{_{[2]}}=|B_{[l+1]}B_{[l]}^{-1}|\cdot|\h{N-2},~\c E^3
\psi(N-1)|_{_{[2]}},\label{1-11}\\
&\prod_{j=0}^{N-3}(c+b_{m+j})\b
f_{_{[2]}}=|B_{[l+1]}B_{[l]}^{-1}|\cdot|\h{N-2},~\c E^3
\psi(N-2)|_{_{[2]}},~\label{1-9}\\
&\prod_{j=1}^{N-2}(c+b_{m+j})\hb
f_{_{[2]}}=|B_{[l+1]}B_{[l]}^{-1}|\cdot|\t{N-1},~\c E^3
\psi(N-1)|_{_{[2]}},~\label{6-21}\\
&\prod_{i=0}^{N-2}(b_{m-1}-a_{n+i})\dh f_{_{[1]}}=|\h{N-2},~\dh \psi(N-1)|_{_{[1]}};\label{1-7}
\end{align}
\end{subequations}
\begin{subequations}\label{3}
\begin{align}
&\!\!(b_{m-1}+a_{n})(a_{n}+c)^{N-2}(b_{m-1}-c)^{N-2}\t{\dh
f}_{_{[3]}}\!\!=\!\!|A_{[n+1]}A_{[n]}^{-1}\!|\cdot|\h{N-3},\dh \psi(N-2),\c E^1\psi(N-2)|_{_{[3]}},\label{2-2}\\
&(a_{n-1}+b_{m})(b_{m}+c)^{N-2}(a_{n-1}-c)^{N-2}\dt{\h
f}_{_{[3]}}=|A_{[m+1]}A_{[m]}^{-1}|\cdot|\h{N-3},\dt \psi(N-2),\c E^2 \psi(N-2)|_{_{[3]}},\label{2-3}\\
&(a_{n-1}-b_{m-1})(a_{n-1}-c)^{N-2}(b_{m-1}-c)^{N-2}\dth
f_{_{[3]}}=|\h{N-3},~\dh \psi (N-2),~\dt \psi
(N-2)|_{_{[3]}},\label{2-1}\\
&(a_{n}+c)^{N-1}\t {\db f}_{_{[3]}}=|A_{[n+1]}A_{[n]}^{-1}
|\cdot|-1,~\h{N-3},~\c E^1\psi(N-2)|_{_{[3]}},\label{6-4}\\
&(a_{n-1}-c)^{N-2}\dt {\b f}_{_{[3]}}=-|\t{N-1},~\dt \psi(N-1)|_{_{[3]}}, \label{6-1}\\
&\!(a_{n-1}+b_m)\!(b_m+c)^{N-2}\!(a_{n-1}-c)^{N-2}\!\dt{\hb
f}_{_{[3]}}=|A_{[m+1]}A_{[m]}^{-1}\!|\cdot|\t{N-2},\!\dt\psi(N-1),\c
E^2\psi(N-1)|_{_{[3]}}\!, \label{6-24}\\
&\prod_{j=1}^{N-2}(a_{n-1}-b_{m+j})\dt {\h f}_{_{[2]}}=
-|\t{N-1},~\dt \psi(N-1)|_{_{[2]}}, \label{6-23}\\
\allowdisplaybreaks[4]&(a_{n-1}-c)\prod_{j=0}^{N-3}(c-b_{m+j})\prod_{j=0}^{N-3}(a_{n-1}-b_{m+j})\db{\dt
f}_{_{[2]}}=|\h{N-3},~\db \psi(N-2),~\dt\psi(N-2)|_{_{[2]}},\label{2-6}\\
&2c\prod_{j=0}^{N-3}(c+b_{m+j})\prod_{j=0}^{N-3}(c-b_{m+j})f_{_{[2]}}=
|B_{[l+1]}B_{[l]}^{-1} |\cdot |\h{N-3},~\db \psi(N-2),~\c
E^3\psi(N-2)|_{_{[2]}},\label{6-8}\\
&2c\prod_{i=1}^{N-2}(c+a_{n+i})\prod_{i=1}^{N-2}(c-a_{n+i})\t
f_{_{[1]}}=|B_{[l+1]}B_{[l]}^{-1} |\cdot |\t{N-2},~\db \psi(N-1),~\c
E^3\psi(N-1)|_{_{[1]}},\label{6-7}\\
&(b_{m-1}-c)\prod_{i=1}^{N-2}(b_{m-1}-a_{n+i})\prod_{i=1}^{N-2}(c-a_{n+i})\t{\dhb
f}_{_{[1]}}=-|\t{N-2},~\dh \psi(N-1),~\db
\psi(N-1)|_{_{[1]}},\label{2-8}\\
&\!\!(a_{n-1}+c)\!\!\prod_{j=0}^{N-3}(c+b_{m+j}\!\!)\prod_{j=0}^{N-3}(a_{n-1}-b_{m+j})\b{\dt
f}_{_{[2]}}\!\!=|B_{[l+1]}B_{[l]}^{-1} \!\!|\cdot|\h{N-3},\dt
\psi(N-2),\c E^3\psi(N-2)|_{_{[2]}},\label{2-7}\\
&(a_{n-1}+c)\!\!\prod_{j=1}^{N-2}\!\!(c+b_{m+j})\!\!\prod_{j=1}^{N-2}\!\!(a_{n-1}-b_{m+j})\hb{\dt
f}_{_{[2]}}\!\!=|B_{[l+1]}B_{[l]}^{-1}\!|\cdot|\t{N-2},\dt
\psi(N-1),\c
E^3\psi(N-1)|_{_{[2]}}, \label{6-22}\\
&(b_{m-1}+c)\!\!\prod_{i=1}^{N-2}\!(c+a_{n+i})\!\!\prod_{i=1}^{N-2}(b_{m-1}-a_{n+i})\t{\b{\dh
f}}_{_{[1]}}\!\!=\!|B_{[l+1]}B_{[l]}^{-1} | \cdot |\t{N-2},\dh
\psi(N-1),\c E^3\psi(N-1)|_{_{[1]}};\label{6-6}
\end{align}
\end{subequations}
\begin{subequations}\label{4}
\begin{align}
&(a_{n-1}-c)^{N-2}[\dt g_{_{[3]}}+(a_{n-1}-c)\dt f_{_{[3]}}]
  =-|\h{N-3},~\psi(N-1),~\dt \psi(N-2)|_{_{[3]}},\label{3-1}\\
&(b_{m}+c)^{N-2}[\h g_{_{[3]}}-(b_{m}+c)\h f_{_{[3]}}]
  =|A_{[m+1]}A_{[m]}^{-1}|\cdot|\h{N-3},\psi(N-1),\c E^2
  \psi(N-2)|_{_{[3]}},\label{3-2}\\
&(a_{n-1}-c)^{N-2}[\dt h_{_{[3]}}+(a_{n-1}-c)\dt g_{_{[3]}}]
  =-|\h{N-3},~\psi(N),~\dt \psi(N-2)|_{_{[3]}},\label{4-1}\\
&(b_{m}+c)^{N-2}[\h h_{_{[3]}}-(b_{m}+c)\h g_{_{[3]}} ]
  =|A_{[m+1]}A_{[m]}^{-1}|\cdot|\h{N-3},~\psi(N),~\c E^2
  \psi(N-2)|_{_{[3]}},\label{4-2}\\
  &(b_{m-1}-c)^{N-2}[\dh g_{_{[3]}}+(b_{m-1}-c)\dh f_{_{[3]}}]
  =-|\h{N-3},~\psi(N-1),~\dh \psi(N-2)|_{_{[3]}},\label{6-15}\\
&(b_{m-1}-c)^{N-2}[(b_{m-1}-c)\dh g_{_{[3]}}+\dh h_{_{[3]}}]
  =-|\h{N-3},~\psi(N),~\dh \psi(N-2)|_{_{[3]}};\label{6-25}
\end{align}
\end{subequations}
\begin{subequations}\label{5}
\begin{align}
&\dh f_{_{[3]}}=\db f_{_{[3]}}-(b_{m-1}-c)|\dh \psi(-1)~~\h{N-2}|_{_{[3]}},\label{6-9}\\
&\dh{\b f}_{_{[3]}}=f_{_{[3]}}-(b_{m-1}-c)g_{_{[3]}}+
(b_{m-1}-c)^2|\dh \psi(-1)~~\t{N-1}|_{_{[3]}},~\label{6-10}\\
&\dh g_{_{[3]}}=|\dh \psi(-1)~~\h{N-2}|_{_{[3]}}-(b_{m-1}-c)
|\dh \psi(-1),~\psi(-1),~\t{N-2}|_{_{[3]}},\label{6-11}\\
&\dt f_{_{[3]}}=\db f_{_{[3]}}-(a_{n-1}-c)|\dt \psi(-1),~\h{N-2}|_{_{[3]}},\label{5-1}\\
&\dt{\b f}_{_{[3]}}=f_{_{[3]}}-(a_{n-1}-c)g+(a_{n-1}-c)^2|\dt \psi(-1),~\t{N-1}|_{_{[3]}},\label{5-2}\\
&\dt g_{_{[3]}}=|\dt \psi(-1),~\h{N-2}|_{_{[3]}}-(a_{n-1}-c)|\dt \psi(-1),~\psi(-1),~\t{N-2}|_{_{[3]}},\label{5-3}\\
&(-1)^N\t f_{_{[3]}}=|A_{[n+1]}A_{[n]}^{-1}|[\db f_{_{[3]}}-(a_{n}+c)|\c E^1 \psi(-1),~\h{N-2}|_{_{[3]}}],\label{5-5}\\
&(-1)^N\t{\b f}_{_{[3]}}=|A_{[n+1]}A_{[n]}^{-1}|[f_{_{[3]}}+(a_n+c)g-(a_{n}+c)^2|\c E^1  \psi(-1),~\t{N-1}|_{_{[3]}}],\label{5-6}\\
&(-1)^N\t g_{_{[3]}}=|A_{[n+1]}A_{[n]}^{-1}|[\!-|\c E^1
\psi(-1),~\h{N-2}|_{_{[3]}}-(a_{n}+c)\!\!|\c E^1
\psi(-1),~\psi(-1),~\t{N-2}|_{_{[3]}}],\label{5-7}\\
&(-1)^N\h f_{_{[3]}}=|A_{[m+1]}A_{[m]}^{-1}
|\cdot[\db f_{_{[3]}}-(b_{m}+c)|\c F^2 \psi(-1),~\h{N-2}|_{_{[3]}}],\label{6-12}\\
&(-1)^N\h{\b f}_{_{[3]}}=|A_{[m+1]}A_{[m]}^{-1}
|\cdot[f_{_{[3]}}+(b_m+c)g-(b_{m}+c)^2|\c E^2\psi(-1),~\t{N-1}|_{_{[3]}}],\label{6-13}\\
&\!\!(-1)^N\h g_{_{[3]}}=\!|A_{[m+1]}A_{[m]}^{-1}|\cdot[-|\c
E^2 \psi(-1),\h{N-2}|_{_{[3]}}-\!(b_{m}+c)|\c
E^2\psi(-1),\psi(-1),\t{N-2}|_{_{[3]}}],\label{6-14}\\
&\!\!(a_{n-1}\!-b_{m-1})\!\dth{\b f}_{_{[3]}}\!=\!\!(a_{n-1}-b_{m-1})\db
f_{_{[3]}}\!-\!\!(a_{n-1}\!-c)^2 |\dt \psi(-1),\h{N-2}|_{_{[3]}}\!\!
+\!\!(b_{m-1}\!-c)^2|\dh\psi(-1),\h{N-2}|_{_{[3]}}\!\!\nonumber\\
&+\!\!(a_{n-1}\!-c)^2(b_{m-1}\!-c)\!|\dt
\psi(-1),\psi(-1),\t{N-2}|_{_{[3]}}\!\!
-(a_{n-1}\!-c)(b_{m-1\!}-c)^2|\dh\psi(-1),\psi(-1),\t{N-2}|_{_{[3]}}\!\!\nonumber\\
&+(a_{n-1}-c)^2(b_{m-1}-c)^2
|\dh\psi(-1),\dt \psi(-1),\t{N-2}|_{_{[3]}},\label{5-4}\\
&(a_n-b_m)\th{\b
f}_{_{[3]}}=|A_{[n+1]}A_{[n]}^{-1}||A_{[m+1]}A_{[m]}^{-1}|[(a_n-b_m)\db
f_{_{[3]}}-(a_n+c)^2|\c
E^1\psi(-1),\h{N-2}|_{_{[3]}}\nonumber\\
&~~+(b_m+c)^2|\c E^2\psi(-1),\h{N-2}|_{_{[3]}}-(a_n+c)^2(b_m+c)|\c
E^1\psi(-1), \psi(-1), \t{N-2}|_{_{[3]}}\nonumber\\
&~~+(a_n\!+c)(b_m\!+c)^2|\c E^2\psi(-1),\psi(-1), \t{N-2}|_{_{[3]}}\nonumber\\
&~~ -(a_n\!+c)^2(b_m+c)^2|\c E^2\psi(-1), \c
E^1\psi(-1), \t{N-2}|_{_{[3]}}].\label{5-8}
\end{align}
\end{subequations}

%===============================================================
In fact, these formulae can be proved in a similar way as in \cite{HZ-PartII,SZ-SIGMA-2010,ZH-H1}.
We need to use the following relations which are derived from \eqref{Appendix-1}:
\begin{subequations}\label{Appendix-2}
\begin{align}
&(a_{n-1}-b_m)\dt\psi=\psi-\dt{\h\psi},~~(b_{m-1}-a_n)\dh\psi=\psi-\dh{\t\psi},\label{Appendix-2-1}\\
&(a_n+c)\t\psi=A_{[n+1]}A_{[n]}^{-1}\psi+\t{\b\psi},~~(b_{m-1}+a_n)\dh{\t\psi}=\t\psi+A_{[n+1]}A_{[n]}^{-1}\dh\psi,
\label{Appendix-2-2}\\
&(c+b_m)\b\psi=B_{[l+1]}B_{[l]}^{-1}\psi+\h{\b\psi},~
(a_{n-1}+c)\b{\dt\psi}=\b\psi+B_{[l+1]}B_{[l]}^{-1}\dt\psi,~2c\psi=\b\psi+B_{[l]}B_{[l-1]}^{-1}\db\psi,\label{Appendix-2-3}\\
&(b_m+c)\h\psi=A_{[m+1]}A_{[m]}^{-1}\psi+\h{\b\psi},~~(a_{n-1}+b_m)\h{\dt\psi}=\h\psi+A_{[m+1]}A_{[m]}^{-1}\dt\psi.\label{Appendix-2-4}
\end{align}
\end{subequations}
In the following as examples we only prove \eqref{1-5} and \eqref{2-7}.
Let us prove \eqref{1-5}. For $\dt
f_{_{[2]}}$,  using the relation \eqref{Appendix-2-1} we first have
\begin{align*}
(a_{n-1}-b_{m})\dt
f_{_{[2]}}&=|(a_{n-1}-b_{m})\dt\psi(0),~\dt\psi(1),~\cdots,~\dt\psi(N-1)|_{_{[2]}}\\
&=|\psi(0),~\dt\psi(1),~\cdots,~\dt\psi(N-1)|_{_{[2]}}.
\end{align*}
Then, for the second column,
\begin{align*}
(a_{n-1}-b_{m+1})(a_{n-1}-b_{m})\dt
f_{_{[2]}}&=|\psi(0),~(a_{n-1}-b_{m+1})\dt\psi(1),~\cdots,~\dt\psi(N-1)|_{_{[2]}}\\
&=|\psi(0),~\psi(1),~\cdots,~\dt\psi(N-1)|_{_{[2]}}.
\end{align*}
Repeating this procedure we reach to
\[
\prod_{j=0}^{N-2}(a_{n-1}-b_{m+j})\dt f_{_{[2]}}=|\h{N-2},~
\dt\psi(N-1)|_{_{[2]}},\]
i.e., \eqref{1-5}.
Next, let us prove \eqref{6-8}. Based on $\b{\eqref{1-6}}$ and using
\eqref{Appendix-2-3}, we first have
\begin{align*}
(c+b_{m})\prod_{j=0}^{N-3}(c-b_{m+j})f_{_{[2]}}
&=-|(c+b_{m})\b\psi(0),~\b\psi(1),~\cdots,~\b\psi(N-2),~
\psi(N-2)|_{_{[2]}}\\
&=-|B_{[l+1]}B_{[l]}^{-1}\psi(0),~\b\psi(1),~\cdots,~\b\psi(N-2),~
\psi(N-2)|_{_{[2]}}.
\end{align*}
For the second column, we have
\begin{align*}
(c+b_{m+1})(c+b_{m})\prod_{j=0}^{N-3}(c-b_{m+j})f_{_{[2]}}&\\
=-|B_{[l+1]}B_{[l]}^{-1}\psi(0)&,~(c+b_{m+1})\b\psi(1),~
\cdots,\b\psi(N-2),~\psi(N-2)|_{_{[2]}}\\
=-|B_{[l+1]}B_{[l]}^{-1}\psi(0)&,~B_{[l+1]}B_{[l]}^{-1}\psi(1),~
\cdots,~\b\psi(N-2),~\psi(N-2)|_{_{[2]}}.
\end{align*}
Repeating this procedure, we reach to
\begin{align*}
&\prod_{j=0}^{N-3}(c+b_{m+j})\prod_{j=0}^{N-3}(c-b_{m+j})f_{_{[2]}}\\
&~~~=-|B_{[l+1]}B_{[l]}^{-1}||\h{N-3},
B_{[l]}B_{[l+1]}^{-1}\b\psi(N-2),
B_{[l]}B_{[l+1]}^{-1}\psi(N-2)|_{_{[2]}}\\
&~~~=|B_{[l+1]}B_{[l]}^{-1}||\h{N-3},
B_{[l]}B_{[l+1]}^{-1}\psi(N-2),\c E^3\psi(N-2)|_{_{[2]}},
\end{align*}
where we have rewritten $B_{[l]}B_{[l+1]}^{-1}\b\psi(N-2)$ by $\c E^3\psi(N-2)|_{_{[2]}}$.
Then, still using \eqref{Appendix-2-3} and the fact that $B_{[l+1]}B_{[l]}^{-1}$ is independent of $l$, we have
\begin{align*}
&2c\prod_{j=0}^{N-3}(c+b_{m+j})\prod_{j=0}^{N-3}(c-b_{m+j})\dt
{\b f}_{_{[2]}}\\
&~~~=|B_{[l+1]}B_{[l]}^{-1}||\h{N-3},
B_{[l]}B_{[l+1]}^{-1}2c\,\psi(N-2),
\c E^3\psi(N-2)|_{_{[2]}}\\
&~~~=|B_{[l+1]}B_{[l]}^{-1}||\h{N-3},B_{[l]}B_{[l+1]}^{-1}B_{[l]}B_{[l-1]}^{-1}\db\psi(N-2),~\c
E^3\psi(N-2)|_{_{[2]}}\\
&~~~=|B_{[l+1]}B_{[l]}^{-1}||\h{N-3},\db\psi(N-2),~\c
E^3\psi(N-2)|_{_{[2]}},
\end{align*}
which is \eqref{6-8}.

\end{appendix}

%%%%%%%%%%%%%%%%%%%%%%%%%%%%%%%%%%%%%%%%%%%%%%%%%%%%%%%%%%%%%%%%%%%%%%%%%%%%%%%%%%%%%%%%%%%%%%%%%%%%%%%%%%%%%%%%%%%%%%%

%\end{CJK*}
\end{document}